\documentclass[journal,onecolumn,draftclsnofoot ]{IEEEtran}
\usepackage{graphicx}
\usepackage{amssymb}
\usepackage{amsthm}
\usepackage{amsmath}
\newtheorem{theorem}{Theorem}
\newtheorem{lemma}{Lemma}
\newtheorem{assumption}{Assumption}
\usepackage{bm}
\usepackage{color}
\usepackage{algorithm}
\usepackage{algorithmic}
\usepackage{subfig}
\usepackage{epstopdf}
\usepackage{psfrag}
\usepackage{booktabs}
\usepackage{lettrine}
\usepackage[numbers,sort&compress]{natbib}

\DeclareMathOperator*{\Min}{minimize}

\newcommand{\reals}{{\mathbb{R}}}

\newcommand{\vx}{{\bf x}}
\newcommand{\vs}{{\bf s}}

\newcommand{\vz}{{\bf z}}

\newcommand{\vy}{{\bf y}}

\newcommand{\vzero}{{\bf 0}}

\newcommand{\vI}{{\bf I}}

\begin{document}
\title{Fast Signal Recovery from Saturated Measurements by Linear Loss and Nonconvex Penalties}

\author{Fan He, Xiaolin Huang, Yipeng Liu, Ming Yan
        % <-this % stops a space
\thanks{
This work is supported by National Natural Science Foundation of China, No. 61603248, 61602091, and NSF grant DMS-1621798.

F. He and X. Huang  (corresponding author) are with Institute of Image Processing and Pattern Recognition, Shanghai Jiao Tong University, and the MOE Key Laboratory of System Control and Information Processing, 200240 Shanghai, P.R. China.
{Y. Liu is with  School of Information and Communication Engineering, University of Electronic Science and Technology of China, Chengdu, 611731, China.}
M. Yan is with the Department of Computational Mathematics, Science and Engineering, Michigan State University, MI, USA.
(e-mails: hf-inspire@sjtu.edu.cn, xiaolinhuang@sjtu.edu.cn, yipengliu@uestc.edu.cn,  yanm@math.msu.edu)}}

\markboth{ }%
{F. He \MakeLowercase{\textit{et al.}}: Fast Signal Recovery from Saturated Measurements by Linear Loss and Nonconvex Penalties}

\maketitle

\begin{abstract}
%Sign information is the key for overcoming the inevitable saturation error in compressive sensing system, which causes loss of information and may result in great bias. To pursue sign consistency, hard constraints and hinge loss are the natural choice for noise-free and noise-corrupted problem, respectively. To further improve the effectiveness of signal recovery from saturation, we in this letter propose to use linear loss instead of hard constraints or hinge loss. With the linear loss, analytical solution in the update progress can be obtained and the estimation error can still be theoretically bounded. Due to the use of linear loss, some non-convex penalties, for which the corresponding subproblem can be analytically solved, are applicable. Generally, with linear loss and those nonconvex penalties, e.g., minimax concave penalty, $\ell_0$ norm, and sorted $\ell_1$ norm, the recovery performance can be significantly improved from the existing methods and the computational time is also largely saved in the numerical experiments.

Sign information is the key to overcoming the inevitable saturation error in compressive sensing systems, which causes information loss and results in bias. For sparse signal recovery from saturation, we propose to use a linear loss to improve the effectiveness from existing methods that utilize hard constraints/hinge loss for sign consistency. Due to the use of linear loss, an analytical solution in the update progress is obtained, and some nonconvex penalties are applicable, e.g., the minimax concave penalty, the $\ell_0$ norm, and the sorted $\ell_1$ norm. Theoretical analysis reveals that the estimation error can still be bounded. Generally, with linear loss and nonconvex penalties, the recovery performance is significantly improved, and the computational time is largely saved, which is verified by the numerical experiments.

\end{abstract}


% \fi
\IEEEpeerreviewmaketitle

\section{Introduction}
\lettrine[lines=2]{S}{aturation} is unavoidable in many sensing systems due to the limited range of detectors or analog-to-digital converters (ADC)~\cite{haboba2012pragmatic}.
When there are saturated measurements, the observation is nonlinear, and the performance of algorithms using linear observations degrades.
We model a measurement from a linear system as
$q_i=\phi_i^\top \bar{\vx} + n_i$
, where $\bar \vx \in \reals^N$ is the true signal, $\phi_i \in \reals^N$ is a sensing vector and $n_i$ is the noise.
Then the observation with a bounded-range detector or ADC becomes
\[
y_i = \max\big\{y_{\mathrm{min}}, \min\{y_{\mathrm{max}}, \phi_i^\top \bar{\vx}+n_i \}\big\},
\]
where $y_{\mathrm{max}}$ and $y_{\mathrm{min}}$ are the upper and lower bounds, respectively.
We partition the observations and the sensing matrix into the unsaturated and saturated parts:
\begin{enumerate}
\item The unsaturated part: observations $\mathbf{y}_1\in\mathbb{R}^{M_1}$ and the corresponding sensing matrix $\mathbf{\Phi}_1\in\mathbb{R}^{M_1\times N}$.
Clearly, $y_{\mathrm{min}} < (\mathbf{y}_1)_i = \left(\mathbf{\Phi}_1\bar\vx\right)_i + n_i < y_{\mathrm{max}}$ for $i=1,\dots,M_1$;
\item The saturated part:
observations $\mathbf{y}_2\in\mathbb{R}^{M_2}$ and the corresponding sensing matrix $\mathbf{\Phi}_2\in\mathbb{R}^{M_2\times N}$.% satisfying
All the observations in this part are out of the range $(y_{\mathrm{min}}, y_{\mathrm{max}})$ and are recorded as $y_{\mathrm{min}}$ or $y_{\mathrm{max}}$.
\end{enumerate}
In addition, an indicator vector $\vs\in\reals^M$ is defined below.
\begin{equation*}
\begin{aligned}
s_i=\left\{\begin{array}{rl}1, & \left(\mathbf{\Phi}\bar\vx\right)_i+n_i  \geq y_{\mathrm{max}},\\
0, & \quad y_{\mathrm{min}} < \left(\mathbf{\Phi}\bar\vx\right)_i+n_i < y_{\mathrm{max}},\\
 -1, & \left(\mathbf{\Phi}\bar\vx\right)_i+n_i \leq y_{\mathrm{min}}. \end{array}\right.
%\begin{cases}
% 1, \quad \left(\mathbf{\Phi}\bar\vx\right)_i+n_i  \geq y_{\mathrm{max}},\\
% 0, \quad y_{\mathrm{min}} < \left(\mathbf{\Phi}\bar\vx\right)_i+n_i < y_{\mathrm{max}},\\
% -1,\quad \left(\mathbf{\Phi}\bar\vx\right)_i+n_i \leq y_{\mathrm{min}}.
% \end{cases}
\end{aligned}
\end{equation*}
Similarly, we partition $\vs$ into two parts $\vs_1$ and $\vs_2$. $s_i \in \vs_1$ when $s_i=0$ and $s_i\in\vs_2$ otherwise.

In this letter, we consider signal recovery from saturated measurements, which is hard due to the loss of information. Compressive sensing (CS,~\cite{donoho2006compressed}) is a promising technique for signal recovery from a relatively small number of observations.
It has been insightfully studied and successfully applied in the last decade~\cite{candes2006robust, candes2006near,candes2008introduction, lustig2007sparse, provost2009application}.
However, the traditional CS is not applicable to deal with saturation.
Two methods for saturation in CS are saturation rejection~\cite{jacques2011dequantizing,laska2009finite} and saturation consistency~{\cite{laska2011democracy,liu2014robust,huang2017mixed,Foucart2017Sparse,Foucart2018Sparse}}.
Saturation rejection drops the saturated part and thus may lead to insufficient measurements and poor results.
To increase the accuracy, we need to make good use of the saturated part.
Although the exact values for the saturated observations are unknown, we can put this information in inequality constraints or loss functions.
This motivates many algorithms, especially in the extreme one-bit case ~\cite{boufounos20081,gupta2010sample, yan2012robust, jacques2013robust}.

The linear inequality constraint 	$(\mathbf{s}_2)_i\left(\mathbf{\Phi}_2\bar\vx-\mathbf{y}_2\right)_i\geq 0$ is obviously true for noise-free cases and thus is used in~\cite{laska2011democracy,liu2014robust} to enforce the saturation consistency.
Specifically, robust dequantized compressive sensing (RDCS) is proposed in~\cite{liu2014robust} and takes the following form
\begin{equation*}
\begin{aligned}
\Min\limits_\vx\quad &\textstyle\nu\|\mathbf{x}\|_1+\frac{1}{2}\|\mathbf{\Phi}_1\mathbf{x}-\mathbf{y}_1\|_2^2\\
\mathrm{s.t.}\quad &(\mathbf{s}_2)_i\left(\mathbf{\Phi}_2\mathbf{x}-\mathbf{y}_2\right)_i\geq0,
\end{aligned}
\end{equation*}
where $\nu > 0$ is a parameter.  (RDCS is for both quantization error and saturation error, of which the former is out of the scope of this letter and hence the corresponding items are ignored here.)
When there are changes of the binary observations, namely sign flips, during the measurement and transmission, the constraints are not satisfied by the true signal.
Therefore, RDCS is not robust to sign flips.
Instead of constraints, the hinge loss function is used in mixed one-bit compressive sensing (M1bit-CS)~\cite{huang2017mixed} to encourage the saturation consistency, and the method is robust to noise.
However, the algorithms for both RDCS and M1bit-CS are slow and not applicable to large-scale problems, due to the use of
%linear
{hard constraints} or the hinge loss. Also, it is less likely to include non-convex penalties that improve the recovery accuracy of sparse signals.

We propose to use the linear loss for sign consistency and develop fast algorithms with non-convex penalty functions for sparse signal recovery.
We formulate it into a form such that
{the alternating direction method of multipliers }
%alternation direction method of multipliers
(ADMM,~\cite{boyd2011distributed}) can be applied.
Then based on the results in~\cite{huang2018nonconvex, zhu2015towards}, the subproblems with non-convex penalties such as the $\ell_0$ norm, the sorted $\ell_1$ norm~\cite{huang2015nonconvex, bogdan2015slope}, and the minimax concave penalty (MCP,~\cite{zhang2010nearly}) have analytical solutions or can be solved easily.

\section{Sparse Signal Recovery from Saturated Measurements}\label{sec:model}
To deal with saturation, we introduce the linear loss instead of the hard constraints and the hinge loss in existing works.
The problem with linear loss is:
\begin{equation}\label{origin}
\begin{aligned}
\Min\limits_{\mathbf{x}\in\mathbb{R}^N}~ &f(\vx)+\frac{1}{2M_1}\|\mathbf{\Phi}_1\mathbf{x}-\mathbf{y}_1\|_2^2-\frac{\gamma}{M_2}\mathbf{s}_2^{\top}(\mathbf{\Phi}_2\mathbf{x}-\mathbf{y}_2)\\
\mathrm{s.t.}~& \|\mathbf{x}\|_2\leq C,
\end{aligned}
\end{equation}
where {$f(\vx)$} is a regularization term for sparsity, $\gamma\geq 0$ is a trade-off parameter between the unsaturated and saturated parts, and $C$ is a given upper bound for $\|\vx\|_2$.
Note that the constraint $\|\mathbf{x}\|_2\leq C$ is crucial when there are many saturated observations, because one-bit information has no capability to distinguish amplitudes.
Here we assume that the $\ell_2$ norm of the true signal is given as $C$.
There are algorithms for estimating the $\ell_2$ norm of the true signal if it is not given~\cite{knudson2016one}.

The key difference between this model and the model in~\cite{huang2017mixed} is the use of linear loss. Thus we call~\eqref{origin} as mixed one-bit CS with linear loss (M1bit-CS-L).
The motivation comes from the good properties of linear loss that it does not bring computational burden in optimization.
Specifically, a subproblem for M1bit-CS-L has a closed-form solution, while that for RDCS and M1bit-CS does not, from which it follows that M1bit-CS-L can be solved as efficiently as standard CS.

Before introducing the algorithms, we show some theoretical results of M1bit-CS-L in terms of the error bound $\|\bar{\vx}-\hat{\vx}\|_2$, where $\hat{\vx}$ is the optimal solution and $\bar{\vx}$ is the true signal.
Assume $\|\bar{\vx}\|_2=1$ and $\mathbb{P}[(\mathbf{\Phi \bar{x}})_i+n_i\notin (y_{\mathrm{min}}, y_{\mathrm{max}})]=p $ without loss of generality.
Then the constraint in~\eqref{origin} is $\|\mathbf{x}\|_2\leq1$.

\begin{assumption}
The true signal $\bar{\vx}$ satisfies that
\begin{equation*}
\begin{aligned}
	\|\mathbf{\Phi \bar{x}}-\vy\|_{\infty}\leq\epsilon\nonumber,
\end{aligned}	
\end{equation*}	
with $\epsilon>0$.
Moreover, the expectation of the noise of the unsaturated part is zero, i.e.,
	\begin{equation*}
	\begin{aligned}
	&\mathbb{E}\left[\mathbf{\Phi}_1\bar{\vx}-\mathbf{y}_1\right]_i=0,\quad i=1,\cdots,M_1.
	\end{aligned}
	\end{equation*}	
\end{assumption}

\begin{assumption}
Each row of $\mathbf{\Phi}$ is an independent realization from a normal distribution,
and the element in $\vs$ is independently drawn at random.

\end{assumption}

Following~\cite{plan2013robust}, we define a function $\eta$ to model the noise in $\vs$ as
\begin{equation*}
	\begin{aligned}
\eta\left((\mathbf{\Phi}\bar{\vx})_i\right) =	\mathbb{E}[s_i|(\mathbf{\Phi}\bar{\vx})_i], \quad i=1,\cdots,M.
	\end{aligned}
\end{equation*}	

\begin{lemma}\label{lemma1} Let
\[
\lambda = \mathbb{E}_{g\sim \mathcal{N}(0,1)}[\eta(g)g],
\] then
\begin{eqnarray*}
\mathbb{E}\left[(\mathbf{\Phi}^{\top})_js_j\right]
& = & \lambda \bar{\vx}, \quad j=1,\cdots,M,\\
\mathbb{E}\left[(\mathbf{\Phi}_2^{\top})_j(\vs_2)_j\right]
& = & \lambda \bar{\vx}/p, \quad j=1,\cdots,M_2.
\end{eqnarray*}
\end{lemma}

Let
$$\sigma=\max\left\{\frac{\epsilon}{M_1}, \frac{\gamma}{M_2}\right\}.$$
Then based on Lemma 1, we prove:
\begin{lemma}\label{lemma2}
Suppose Assumptions 1 and 2 hold and $t$ is given.
With probability at least $1-e^{1-t}$, the following holds:
\begin{equation*}
\begin{aligned}
\left\|\frac{\gamma}{M_2}\mathbf{\Phi}_2^{\top}\mathbf{s}_2-\frac{1}{M_1}\mathbf{\Phi}_1^{\top}\left(\mathbf{\Phi}_1\bar{\vx}-\mathbf{y}_1\right)-\frac{\gamma\lambda\bar{\vx}}{p}\right\|_{\infty}~~~~~\\
 \leq c\sqrt{\sigma M(t+\log N)}\triangleq
\frac{\nu}{2}.
\end{aligned}
\end{equation*}
\end{lemma}

The proofs of Lemma 1 and Lemma 2 are presented in the supplemental document. Based on these lemmas, we bound the error for M1bit-CS-L by the following theorem.

\begin{theorem}\label{theorem}
Suppose Assumptions 1 and 2 hold, $\hat{\vx}$ is the optimal solution of (\ref{origin}), and $\bar{\vx}$ is the underlying signal.

If $f(\mathbf{x})=\nu\|\mathbf{x}\|_1$, with probability at least $1-e^{1-t}$,
\begin{equation*}
\begin{aligned}
 &\|\bar{\vx}-\hat{\vx}\|_2 \leq \frac{3p\nu}{\gamma\lambda}\sqrt{\|\bar{\vx}\|_0}	= {\mathcal{O}\left(\sqrt{(\sigma MK\log N)/\gamma}\right)
    \nonumber,}
\end{aligned}
\end{equation*}

If $f(\mathbf{x})=\nu\|\mathbf{x}\|_0$, with probability at least $1-e^{1-t}$,
\begin{equation*}
\begin{aligned}
&\|\bar{\vx}-\hat{\vx}\|_2\leq
\sqrt{\frac{4p\nu}{\gamma\lambda}\|\bar{\vx}\|_0}	= {\mathcal{O}\left(\sqrt[4]{(\sigma MK^2\log N)/\gamma^2}\right)}
    \nonumber.
\end{aligned}
\end{equation*}
\end{theorem}

\begin{proof}
Since {$\bar{\vx}$} is a feasible solution of~\eqref{origin} and {$\hat{\vx}$} is the optimal solution of~\eqref{origin}, we obtain:
\begin{equation*}
\begin{aligned}
0\geq&  f\left(\hat{\vx}\right)+\frac{1}{2M_1}\|\mathbf{\Phi}_1\hat{\vx}-\vy_1\|_2^2-\frac{\gamma}{M_2}\vs_2^{\top}\left(\mathbf{\Phi}_2\hat{\vx}-\vy_2\right)\\
& -f\left(\bar{\vx}\right)-\frac{1}{2M_1}\|\mathbf{\Phi}_1\bar{\vx}-\vy_1\|_2^2+\frac{\gamma}{M_2}\vs_2^{\top}\left(\mathbf{\Phi}_2\bar{\vx}-\vy_2\right)\nonumber.
\end{aligned}
\end{equation*}
The convexity of $\frac{1}{2M_1}\|\mathbf{\Phi}_1\vx-\vy_1\|_2^2$ gives:
\begin{equation*}
\begin{aligned}
0\geq&
%   f(\hat{\vx})-f(\bar{\vx})
%	-\frac{\gamma}{M_2}\vs_2^{\top}\mathbf{\Phi}_2\left(\hat{\vx}-\bar{\vx}\right)\\
%	&+\frac{1}{M_1}\left(\mathbf{\Phi}_1\bar{\vx}-\vy_1\right)^{\top}\mathbf{\Phi}_1\left(\hat{\vx}-\bar{\vx}\right)	\\	
%=&
    f(\hat{\vx})-f(\bar{\vx})
    +\big\langle \frac{\gamma}{M_2}\mathbf{\Phi}_2^{\top}\vs_2-\frac{1}{M_1}\mathbf{\Phi}_1^{\top}\left(\mathbf{\Phi}_1\bar{\vx}-\vy_1\right),\bar{\vx}-\hat{\vx} \big\rangle\\	
 =&f(\hat{\vx})-f(\bar{\vx})
	+\left\langle \frac{\gamma\lambda\bar{\vx}}{p}, \bar{\vx}- \hat{\vx} \right\rangle\\
	&+\big\langle \frac{\gamma}{M_2}\mathbf{\Phi}_2^{\top}\vs_2-\frac{1}{M_1}\mathbf{\Phi}_1^{\top}\left(\mathbf{\Phi}_1\bar{\vx}-\vy_1\right)-\frac{\gamma\lambda\bar{\vx}}{p}, \bar{\vx}- \hat{\vx} \big\rangle\\
\geq& f(\hat{\vx})-f(\bar{\vx})
	+\frac{\gamma\lambda}{p}\left(1-\bar{\vx}^{\top}\hat{\vx}\right)\\
	&-\left\|\frac{\gamma}{M_2}\mathbf{\Phi}_2^{\top}\vs_2-\frac{1}{M_1}\mathbf{\Phi}_1^{\top}\left(\mathbf{\Phi}_1\bar{\vx}-\vy_1\right)-\frac{\gamma\lambda\bar{\vx}}{p}\right\|_{\infty}\|\bar{\vx}- \hat{\vx}\|_1\\
\geq& f(\hat{\vx})-f(\bar{\vx})
	+\frac{\gamma\lambda}{p}\left(1-\bar{\vx}^{\top}\hat{\vx}\right)
	-\frac{\nu}{2}\|\bar{\vx}- \hat{\vx}\|_1\nonumber,
\end{aligned}
\end{equation*}
where the last inequality follows from Lemma~\ref{lemma2}.

In the following, we define $T$ as the support set of $\bar \vx$ and $T^c$ be its complement, i.e. $T^c=\{1,2,\cdots,N\}\backslash T$.

First, let $f(\mathbf{x})=\nu\|\mathbf{x}\|_1$, and we have
\begin{equation}\label{analysis-1}
\begin{aligned}
\frac{\gamma\lambda}{p}\left(1-\bar{\vx}^{\top}\hat{\vx}\right) \leq \nu\|\bar{\vx}\|_1-\nu\|\mathbf{\hat{\vx}}\|_1
	+\frac{\nu}{2}\|\bar{\vx}- \hat{\vx}\|_1.
\end{aligned}
\end{equation}
Thus, using $T$ and $T^c$, we obtain
\begin{equation*}
\begin{aligned}
&\frac{\gamma\lambda}{p}\left(1-\bar{\vx}^{\top}\hat{\vx}\right)\\
\leq&\nu\|\bar{\vx}_{T}\|_1-\nu\|\mathbf{\hat{\vx}}_{T}\|_1-\nu\|\mathbf{\hat{\vx}}_{T^c}\|_1
	+\frac{\nu}{2}\|\bar{\vx}_{T}- \hat{\vx}_{T}\|_1+\frac{\nu}{2}\|\hat{\vx}_{T^c}\|_1\\
	\leq& \frac{3\nu}{2}\|\bar{\vx}_{T}- \hat{\vx}_{T}\|_1
	-\frac{\nu}{2}\|\hat{\vx}_{T^c}\|_1
\leq \frac{3\nu}{2}\|\bar{\vx}- \hat{\vx}\|_2\sqrt{\|\bar{\vx}\|_0}.
\end{aligned}
\end{equation*}
Therefore,
\begin{equation*}
\begin{aligned}
\|\bar{\vx}-\hat{\vx}\|_2^2
 \leq2\left(1-\bar{\vx}^{\top}\hat{\vx}\right)
\leq \frac{3p\nu}{\gamma\lambda}\|\bar{\vx}- \hat{\vx}\|_2\sqrt{\|\bar{\vx}\|_0},
\nonumber
\end{aligned}
\end{equation*}
which implies $\|\bar{\vx}-\hat{\vx}\|_2 \leq 3p\nu\sqrt{\|\bar{\vx}\|_0}/(\gamma\lambda)$.

Next, let $f(\mathbf{x})=\nu\|\mathbf{x}\|_0$ with the same $\nu$ in Lemma~\ref{lemma2}.
We have a similar inequality to~\eqref{analysis-1}:
\begin{equation*}
\begin{aligned}
\frac{\gamma\lambda}{p}\left(1-\bar{\vx}^\top\hat{\vx}\right) \leq &\nu\|\bar{\vx}\|_0-\nu\|\mathbf{\hat{\vx}}\|_0
	+\frac{\nu}{2}\|\bar{\vx}- \hat{\vx}\|_1.
\end{aligned}
\end{equation*}
From the definiton of $\bar\vx$ and $\hat\vx$, we have $\|\bar{\vx}-\hat{\vx}\|_2\leq 2$ and
\begin{equation*}
\begin{aligned}
&\frac{\gamma\lambda}{p}\left(1-\bar{\vx}^\top\hat{\vx}\right)
\leq\nu\|\bar{\vx}\|_0-\nu\|\mathbf{\hat{\vx}}\|_0
	+\nu\|\bar{\vx}- \hat{\vx}\|_0\\	
\leq&\nu\|\bar{\vx}_{T}\|_0-\nu\|\mathbf{\hat{\vx}}_{T}\|_0-\nu\|\mathbf{\hat{\vx}}_{T^c}\|_0
	+\nu\|\bar{\vx}_{T}- \hat{\vx}_{T}\|_0+\nu\|\hat{\vx}_{T^c}\|_0\\	
\leq&2\nu\|\bar{\vx}_{T}- \hat{\vx}_{T}\|_0
\leq2\nu\|\bar{\vx}\|_0\nonumber.
\end{aligned}
\end{equation*}
Thus, we obtain
\begin{equation*}
\begin{aligned}\textstyle
\|\bar{\vx}-\hat{\vx}\|_2^2
 \leq2\left(1-\bar{\vx}^{\top}\hat{\vx}\right)
\leq \frac{4p\nu}{\gamma\lambda}\|\bar{\vx}\|_0\nonumber,
\end{aligned}
\end{equation*}
which implies that $\|\bar{\vx}-\hat{\vx}\|_2\leq\sqrt{{4p\nu\|\bar{\vx}\|_0}/({\gamma\lambda})}$.
Then the theorem is proved.
\end{proof}

\section{Fast Algorithms for Convex and Non-convex Penalties}\label{sec:alg}

In this section, we design fast algorithms for both convex and non-convex penalties in the framework of the alternating direction method of multipliers (ADMM) to solve~\eqref{origin}. Introducing an auxiliary vector $\vz$ and an additional constraint $\vx-\vz=\vzero$, we have the following equivalent problem of~\eqref{origin}:
\begin{equation*}
\begin{aligned}
\Min\limits_{\vx,\vz}~&  f(\mathbf{z})-\frac{\gamma}{M_2}\mathbf{s}_2^{\top}\left(\mathbf{\Phi}_2\mathbf{z}-\mathbf{y}_2\right)+\frac{1}{2M_1}\|\mathbf{\Phi}_1\mathbf{x}-\mathbf{y}_1\|_2^2\\
\mathrm{s.t.}~& \textstyle \|\mathbf{z}\|_2\leq C,\quad\quad\mathbf{x}-\mathbf{z}=\mathbf{0}.
\end{aligned}
\end{equation*}
The corresponding augmented Lagrangian is
\begin{equation*}
\begin{aligned}
\mathcal{L}(\mathbf{x},\mathbf{z};&~\textstyle \bm{\alpha})=f(\mathbf{z})-\frac{\gamma}{M_2}\mathbf{s}_2^{\top}\left(\mathbf{\Phi}_2\mathbf{z}-\mathbf{y}_2\right)+\mathbb{I}_C(\mathbf{z})\\
& \textstyle +\frac{1}{2M_1}\|\mathbf{\Phi}_1\mathbf{x}-\mathbf{y}_1\|_2^2+\bm{\alpha}^{\top}\left(\mathbf{x}-\mathbf{z}\right)+\frac{\rho}{2}\|\mathbf{x}-\mathbf{z}\|_2^2,
\end{aligned}
\end{equation*}
where $\mathbb{I}_C(\mathbf{z})$ is the indicator function
%returns
{returning} 0 if $\|\mathbf{z}\|_2\leq C$ and $+\infty$ otherwise.
Then we establish the following two subproblems to update $\vx$ and $\vz$, respectively.
\begin{enumerate}
	\item
	$\mathbf{x}$-subproblem:
	\begin{align*}
	\textstyle \Min\limits_\mathbf{x}~\frac{1}{2M_1}\left\|\mathbf{\Phi}_1\mathbf{x}-\mathbf{y}_1\right\|_2^2+\bm{\alpha}^\top\mathbf{x}+\frac{\rho}{2}\mathbf{x}^{\top}\mathbf{x}-\rho\mathbf{z}^{\top}\mathbf{x}\nonumber.
	\end{align*}
	It is a quadratic problem, and its solution is
	\begin{align*}
	\textstyle \mathbf{x}=\left(\frac{1}{M_1}\mathbf{\Phi}_1^{\top}\mathbf{\Phi}_1+\rho\vI\right)^{-1}\left(\frac{1}{M_1}\mathbf{\Phi}_1^\top\mathbf{y}_1-\bm{\alpha}+\rho \mathbf{z}\right)\nonumber.
	\end{align*}
	
	\item
	$\mathbf{z}$-subproblem:
	\begin{equation*}
	\begin{aligned}
\Min\limits_\mathbf{z}~&\textstyle f(\mathbf{z})-\frac{\gamma}{M_2}\mathbf{s}_2^{\top}\mathbf{\Phi}_2\mathbf{z}+\mathbb{I}_C(\mathbf{z})-\bm{\alpha}^{\top}\mathbf{z}\\
	&\textstyle +\frac{\rho}{2}\mathbf{z}^{\top}\mathbf{z}-\rho\mathbf{x}^{\top}\mathbf{z}\nonumber,
	\end{aligned}
	\end{equation*}
	which can be reformulated as:
	\begin{equation}
	\begin{aligned}	
	%\Min\limits_\vz\quad&f(\vz)+\frac{\rho+\mu}{2}\left\|\vz-\frac{\frac{\gamma}{M_2}\mathbf{\Phi}_2^{\top}\vs_2+\bm{\alpha}+\rho\vx}{\rho+\mu}\right\|_2^2\nonumber\\
	\Min\limits_\vz~&\textstyle f(\vz)- \left\langle \frac{\gamma}{M_2}\mathbf{\Phi}_2^{\top}\vs_2+ \bm{\alpha}+ \rho \vx, \vz\right\rangle +{\rho\over2}\|\vz\|_2^2\\
	\mathrm{s.t.} ~&\textstyle \|\vz\|_2\leq C\nonumber.
	\end{aligned}
	\end{equation}
\end{enumerate}

According to~\cite{huang2018nonconvex}, analytical solutions exist for many convex and nonconvex penalty functions.
Some examples are:
\begin{itemize}
\item {\textbf{$\ell_1$ norm (L1):}} $f(\vz)=\nu\|\vz\|_1$.
\item {\textbf{$\ell_0$ penalty (L0):}} $f(\vz)=\nu\|\vz\|_0$.
\item {\textbf{minimax concave penalty (MCP):}} $f(\mathbf{z})=\sum_{i=1}^Ng_{\nu,b}(z_i)$, where $g_{\nu,b}(z)$ is defined as:
	\begin{equation*}
	\begin{aligned}
	g_{\nu,b}(z)=\begin{cases}
	\nu|z|-z^2/(2b),\quad &\mathrm{if\;}|z|\leq b\nu,\\
	b\nu^2/2,\quad &\mathrm{if\;}|z|> b\nu.
	\end{cases}\nonumber
	\end{aligned}
	\end{equation*}
	
\item {\textbf{nonconvex sorted $\ell_1$ norm (sL1):}} $f(\mathbf{z})=\nu\sum_{i=1}^Nw_i|z_{[i]}|$, where $w_1\geq \dots\geq w_N\geq 0$, and $\{z_{[i]}\}_{i=1}^N$ is a permutation of $\{z_i\}_{i=1}^N$ such that $|z_{[1]}|\leq\dots\leq |z_{[N]}|$.

\end{itemize}
More examples can be found in~\cite{huang2018nonconvex}.

The ADMM algorithm is described in Alg.~\ref{algorithm1}.
\renewcommand{\algorithmicrequire}{\textbf{Input}}
\renewcommand{\algorithmicensure}{\textbf{Output}}
\begin{algorithm}[H]
\begin{algorithmic}[1]
\caption{ADMM for M1bit-CS-L}
\label{algorithm1}
\REQUIRE
$\mathbf{\Phi}_1,~\mathbf{y}_1,~\mathbf{\Phi}_2,~\mathbf{y}_2, C, \gamma$
\ENSURE $\mathbf{x}$
\STATE  initialize $\mathbf{z}=\vzero,~\bm{\alpha}=\mathbf{0},~\rho>0$
\REPEAT
     \STATE $\mathbf{x}:=(\mathbf{\Phi}_1^{\top}\mathbf{\Phi}_1/M_1+\rho)^{-1}(\mathbf{\Phi}_1\mathbf{y}_1/M_1-\bm{\alpha}+\rho\mathbf{z})$
     \STATE Solve the $\vz$-subproblem using~\cite{huang2018nonconvex}
     \STATE $\bm{\alpha}:=\bm{\alpha}+\rho(\mathbf{x}-\mathbf{z})$
 \UNTIL{the stopping criteria is satisfied}
\end{algorithmic}
\end{algorithm}

%ADMM is also applied to in \cite{liu2014robust} and \cite{huang2017mixed} for RDCS and M1bit-CS, respectively. Since the corresponding subproblems with both hard constraints and hinge losses do not have analytical update. Therefore, the algorithm for M1bit-CS-L, even with non-convex penalties, is much faster than that for both RDCS and M1bit-CS.
{ADMM is also applied in [11] and [12] for RDCS and M1bit-CS, respectively, but since the corresponding subproblems with both hard constraints and hinge losses do not have analytical updates, the algorithm for M1bit-CS-L, even with non-convex penalties, is much faster than both RDCS and M1bit-CS.}

\section{Numerical Experiments}\label{sec:num}
We generate the data in the following steps:
i) generate a $K$-sparse signal in $\mathbb{R}^N$ with nonzero components following the normal distribution;
ii) normalize the signal such that $\|\bar{\vx}\|_2 = 1$;
iii) generate a sensing matrix $\mathbb{R}^{M\times N}$ with elements following $\mathcal{N}(0,1)$ independently;
iv) add independent Gaussian noise to the $M$ measurements, where the noise level $s_n$ is the ratio of the variance of the noise to that of the noise-free measurements;
v) set the saturation thresholds such that $M_2$ measurements are saturated.

The task is to recover the sparse signal $\bar{\vx}$ from $M_1$ unsaturated and $M_2$ saturated measurements.
The proposed Alg.~\ref{algorithm1} with four sparse penalties discussed previously will be compared with RDCS, M1bit-CSC, and LASSO~\cite{Tibshirani2011Regression}.
%We use LASSO on the $M_1$ unsaturated measurements only as the baseline.
There is a parameter to tune the sparsity in each {algorithm}.
%We set the parameters properly such that the results have similar sparsity.
%Thus, parameters for RDCS, M1bit-CSC, LASSO, Alg.1-L1 (they all use $\ell_1$ norm ) keep the same.
{For RDCS, M1bit-CSC, LASSO, and Alg.1-L1, we use the same value, which is chosen by cross-validation based on LASSO. For other methods, we choose parameters to make the numbers of non-zero compontents are no more than that of  $\ell_1$-norm minimization. }
The signal-to-noise ratio (SNR) and the angular error (AE) are used as error metrics, and they are defined as
\begin{align*}
\mathrm{SNR}(\bar{\vx}, \hat{\vx})&=10\log_{10}\left({\|\bar{\vx}\|_2^2}/{\|\bar{\vx}-\hat{\vx}\|^2_2}\right),\\
\mathrm{AE}&=\mathrm{arccos}\left({\langle \bar{\vx}, \hat{\vx} \rangle}/({\|\bar{\vx}\|_2\|\hat{\vx}\|_2})\right)/{\pi},
\end{align*}
where $\bar{\vx}$ is the true signal and $\hat{\vx}$ is the recovered one. All the experiments are conducted on Matlab R2016b in Windows 7 with Core i5-3.20 GHz and 4.0 GB RAM.

First, we set $K=100,~N=1000,~M=500,~ s_n=10$ and vary the saturation ratio from $0\%$ to $50\%$. The average SNR and AE over 100 trials are shown in Fig. \ref{K=100N=1000M=500}.
The disadvantage of LASSO is obvious as the saturated ratio increases.
Once the number of unsaturated measurement is insufficient, the accuracy dramatically drops.
%On the contrast,
{In contrast}, the declines in accuracy of other methods are relatively slow due to the knowledge obtained from saturation measurements.
The signal recovery accuracy of M1bit-CSC and Alg.1-L1 are similar, which coincides with our analysis in Section II that using linear loss has little negative impact on performance.
However, as shown in Table \ref{table_time},  Alg.~1 is generally 10 times faster than RDCS and M1bit-CSC because of the analytical update.

\begin{figure}[!htb]
\centering
\subfloat[]{\includegraphics[width=0.45\textwidth]{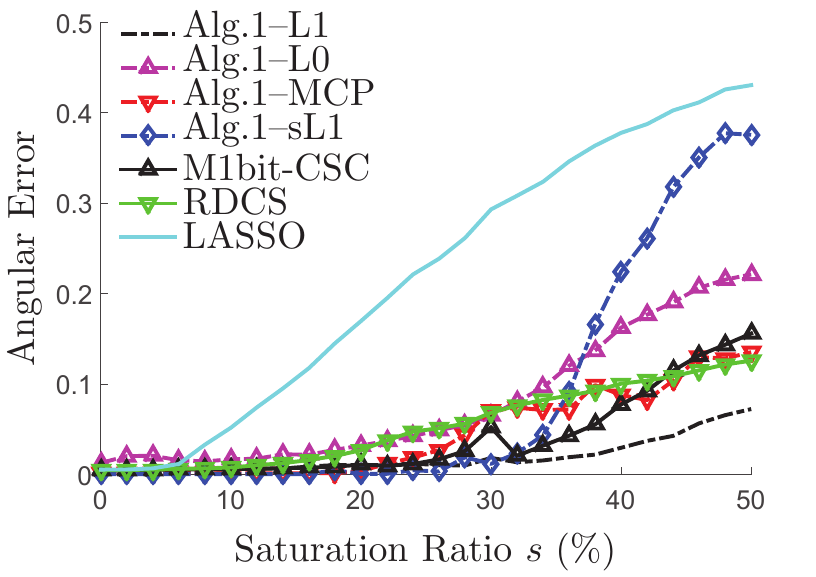}
\label{K=100N=1000M=500-AE}}
%\hfil
\subfloat[]{\includegraphics[width=0.45\textwidth]{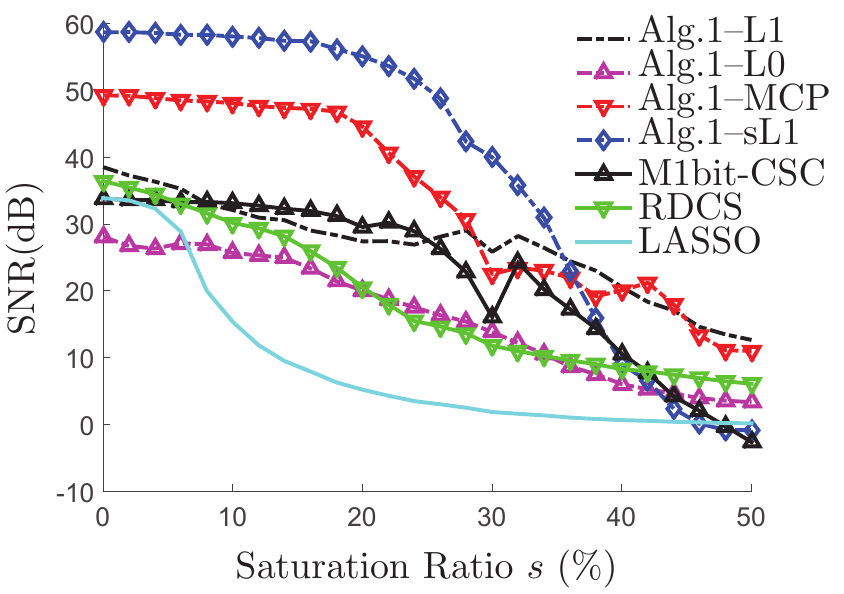}
\label{K=100N=1000M=500-SNR}}
\caption{
(a) AE and (b) SNR averaged over 100 trials for different saturation ratio ($K=100,~N=1000,~M = 500,~s_n=10$).}
\label{K=100N=1000M=500}
\end{figure}

\begin{table}[htbp]
\renewcommand{\arraystretch}{1.3}
\newcommand{\tabincell}[2]{\begin{tabular}{@{}#1@{}}#2\end{tabular}}
\caption{Average computational time with different $M,\;N$ when $K=100,~s_n = 10,~s = 10\%$.}
\label{table_time}
\centering
\begin{tabular}{r|r|r|r|r|r}
\toprule
\bfseries Methods
& \tabincell{c}{\bfseries M=500\\ \bfseries N=1000}
& \tabincell{c}{\bfseries M=1000\\ \bfseries N=1000}
& \tabincell{c}{\bfseries M=500\\ \bfseries N=2000}
& \tabincell{c}{\bfseries M=1000\\ \bfseries N=2000}
& \tabincell{c}{\bfseries M=1500\\ \bfseries N=2000}\\
\midrule
\bfseries LASSO & 0.0121 s& 0.0436 s& 0.0296 s& 0.1333 s& 0.1563 s\\
%\hline
\bfseries RDCS & 0.9355 s& 0.5129 s& 8.5300 s& 7.9630 s& 5.5730 s\\
%\hline
\bfseries M1bit-CSC & 0.9627 s& 1.0500 s& 8.5410 s& 9.2600 s& 8.7560 s\\
%\hline
\bfseries Alg.1--sL1 & 0.0929 s& 0.1340 s& 0.3713 s& 0.5177 s& 0.7089 s\\
%\hline
\bfseries Alg.1--MCP & 0.1306 s& 0.1663 s& 0.5907 s& 0.7127 s& 0.7265 s\\
%\hline
\bfseries Alg.1--L0 & 0.1073 s& 0.1430 s& 0.5604 s& 0.6758 s& 0.6958 s\\
%\hline
\bfseries Alg.1--L1 & 0.1004 s& 0.1375 s& 0.5548 s& 0.6671 s& 0.6854 s\\
\bottomrule
\end{tabular}
\end{table}

With the use of nonconvex penalties, the reconstruction performance is significantly improved by enhancing the sparsity, which is shown in Fig. 2 with $K=100,~N=1000,~s_n=10,~s=15\%$ and changing number of measurements.
Alg.1-sL1 and Alg.1-MCP outperform other methods on both AE and SNR except in the case of insufficient $M$, where no method can recover reasonable signals.
The $\ell_0$-norm is the true sparsity measurement.
However, its optimization
is easy to be trapped in a bad local optimum.
Thus, the average performance of Alg.1-L0 is not very good.
In Fig. 2, one can also observe the effectiveness of using saturated information.
For example, when the number of total measurements is 800, with saturation ratio being 15\%, LASSO actually uses 680 unsaturated measurements.
However, with 680 measurements, including both saturation ones and unsaturated ones, other models give comparable results.

\begin{figure}[!htb]
\centering

\subfloat[]{\includegraphics[width=0.45\textwidth]{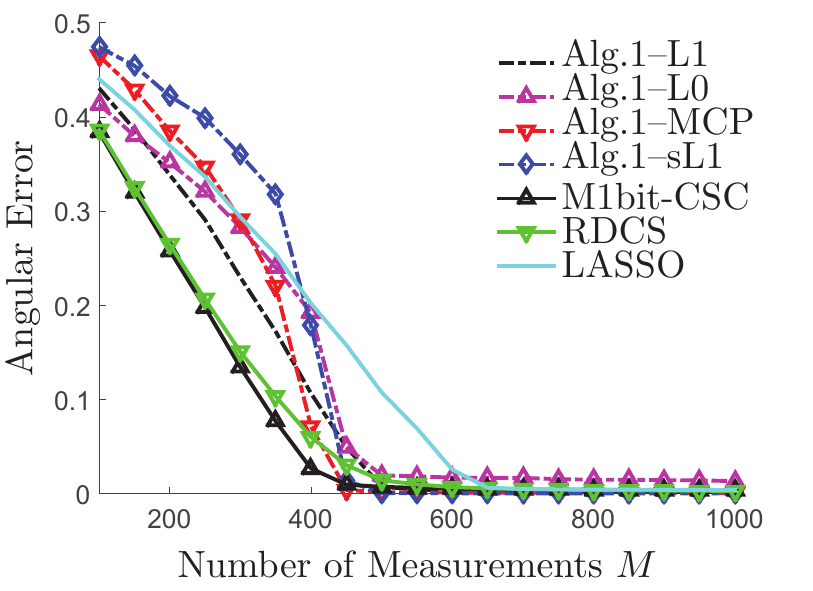}\label{N=1000K=100S=0.15-AE}}
%\hfil
\subfloat[]{\includegraphics[width=0.45\textwidth]{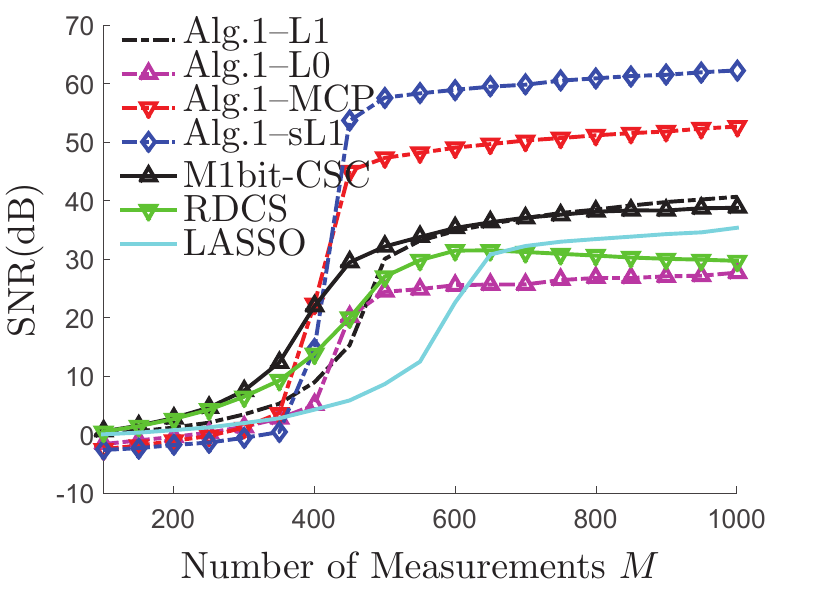}
\label{N=1000K=100S=0.15-SNR}}
\caption{ (a) AE and (b) SNR averaged over 100 trials for different numbers of measurements ($N=1000,~K=100,~s_n = 10,~s = 15\%$).}
\label{N=1000K=100S=0.15}
\end{figure}

\iftrue
Last, we evaluate the performance of these methods when the sparsity changes in Fig.~\ref{S=0.15N=1000M=500} with $N=1000,\;M=500,\;s_n=10,\;s=15\%$.
The performance again confirms the effectiveness of the proposed algorithms: i) the use of linear loss for one-bit information does not decrease the reconstruction performance; ii) the use of suitable non-convex penalties does improve the reconstruction quality.

\begin{figure}[!htb]
\centering

\subfloat[]{\includegraphics[width=0.45\textwidth]{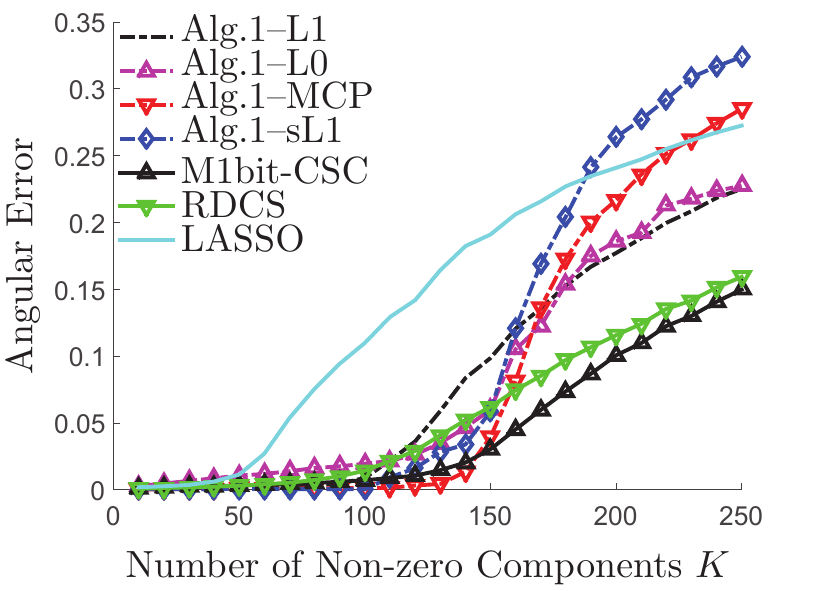}
\label{S=0.15N=1000M=500-AE}}
%\hfil
\subfloat[]{\includegraphics[width=0.45\textwidth]{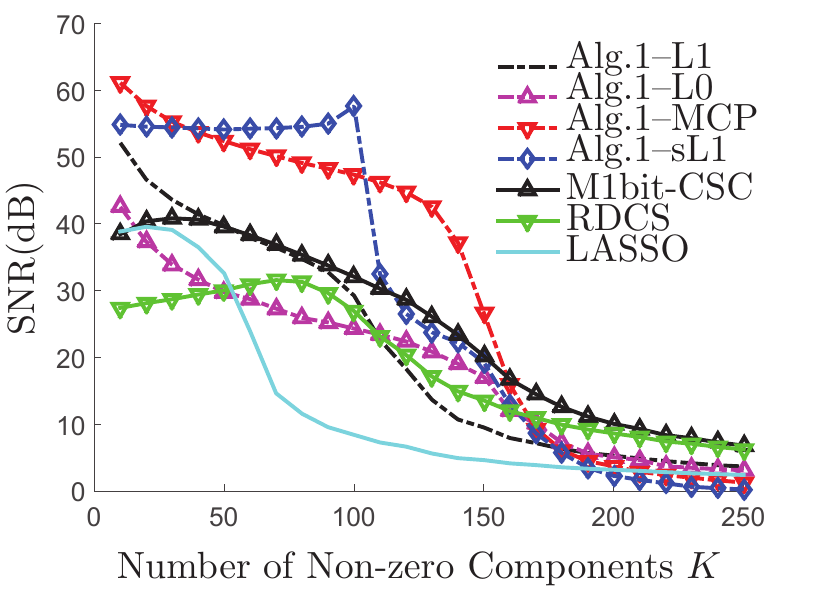}
\label{S=0.15N=1000M=500-SNR}}
\caption{(a) AE and (b) SNR averaged over 100 trials for different numbers of non-zero components ($N=1000,~M = 500,~s_n = 10,~s = 15\%$).}
\label{S=0.15N=1000M=500}
\end{figure}
\fi

\section{Conclusion}\label{sec:con}
To recover sparse signal from sensing systems with saturation, the information contained in the saturated part is very important.
We
%proposed
{propose} minimizing the linear loss for saturation consistency.
Linear loss can be efficiently minimized by the proposed algorithm, and it allows the use of non-convex penalties to further enhance the sparsity.
The error estimation given in this letter also theoretically guarantees the good performance of linear loss.
Numerical experiments indicate the good performance of the proposed method on both accuracy and efficiency.

\newpage

%\appendix

%\section*{Acknowledgment}


\ifCLASSOPTIONcaptionsoff
  \newpage
\fi



%\bibliographystyle{IEEEtran}


%\bibliographystyle{IEEEtran}
%\bibliography{IEEEabrv,M1bitCSlinear}

\end{document}